 \documentclass[11pt]{llncs}
 \usepackage{fullpage}

\usepackage{amssymb,amsmath,amsfonts,bm}

\usepackage{mathtools}
\usepackage{subfigure}
\usepackage{epic}
\usepackage{eepic}
\usepackage{enumerate}
\usepackage{vinayak}
\usepackage{graphicx}
\usepackage{verbatim}
\usepackage{booktabs}

\usepackage{amsthm}
\usepackage[mathscr]{eucal}
\usepackage{xspace}
\usepackage{stmaryrd}
\usepackage{paralist}
\usepackage[norelsize,linesnumbered,ruled]{algorithm2e}

\usepackage{tikz}
\usetikzlibrary{arrows,backgrounds,decorations,decorations.pathmorphing,positioning,fit,automata,shapes,snakes,patterns,plotmarks}

\makeatletter
\let\oldabs\abs
\def\abs{\@ifstar{\oldabs}{\oldabs*}}
\let\oldnorm\norm
\def\norm{\@ifstar{\oldnorm}{\oldnorm*}}
\makeatother

\makeatletter
\g@addto@macro \normalsize {%
 \setlength\abovedisplayskip{5pt plus 2pt minus 2pt}%
 \setlength\belowdisplayskip{4pt plus 2pt minus 2pt}%
}
\makeatother

\begin{document}

\pagestyle{empty}

\title{
Quantifying  Conformance using  the Skorokhod Metric
\thanks{\small This research was funded in part  by
a Humboldt foundation grant,
FCT grant SFRHBPD902672012,
and by a contract from Toyota Motors.}
\vspace{-5mm}
}
\author{
 Jyotirmoy V. Deshmukh$^1$ \and Rupak Majumdar$^2$ \and Vinayak S. Prabhu$^{2,3}$}
\institute{ 
$^1$ Toyota Technical Center;
 \quad
$^2$ MPI-SWS
\quad
$^3$ University of Porto\\
{\tt jyotirmoy.deshmukh@tema.toyota.com}\quad
{\tt $\{$rupak,vinayak$\}$@mpi-sws.org}
}

\maketitle

\begin{abstract}
The conformance testing problem for dynamical systems asks, given two dynamical models (e.g., as Simulink diagrams),
whether their behaviors are ``close'' to each other.
In the semi-formal approach to conformance testing, the two systems are simulated on a large set of tests, 
and a metric, defined on pairs of real-valued, real-timed trajectories, is used to determine a lower bound on the distance. 
We show how the Skorkhod metric on continuous dynamical systems can be used as the
foundation for conformance testing of complex dynamical models.
The Skorokhod metric allows for both state value mismatches and 
timing distortions, and is thus well suited for  checking conformance between
idealized models of dynamical systems and their implementations.
We demonstrate the  robustness of the   system conformance quantification 
by proving a  \emph{transference theorem}: 
trajectories close under the Skorokhod metric satisfy ``close'' logical properties. 
Specifically, we show the result for the timed linear time logic
 \TLTL  augmented with
a rich class of temporal and spatial constraint predicates.
We provide a window-based streaming algorithm to compute the Skorokhod metric, and use it
as a basis for a conformance testing tool for Simulink.
We experimentally demonstrate the effectiveness of  our  tool in finding discrepant behaviors
on a set of control system benchmarks, including an industrial challenge problem.

\end{abstract}

\section{Introduction}
\vspace{-2mm}
A fundamental question in model-based design is {\em conformance testing}: whether two models of a system
are equivalent.
For discrete systems, this question is 
well-studied~\cite{Milner80,HennessyM85,HHK95,BisimBook2011},
and there is a rich theory of process equivalences based on similarity and bisimilarity.
For continuous and hybrid systems, however, the state of the art is somewhat unsatisfactory.
While there is a straightforward generalization of process equivalences to the continuous case,
in practice, equivalence notions such as bisimilarity are always too strong and most systems are not
bisimilar.
Since equivalence is a Boolean notion, one gets no additional information about the systems other than they are ``not bisimilar,''
and even if two dynamical systems are bisimilar, they may still differ in many properties that are of control-theoretic
interest.
Thus, classical notions for equivalence and conformance have been of limited use in industrial practice.

In recent years, the notion of bisimulation has therefore been generalized to \emph{metrics} on systems,
which quantify the distance between them.
For example, one approach is that of $\epsilon$-bisimulation,
which requires that the states of the two systems remain ``close'' forever (within an $\epsilon$-ball), rather than coincide exactly.
Under suitable stability assumptions on the dynamics, one can prove results about $\epsilon$-bisimulation \cite{GirardPT10,HaghverdiTP05}.
Unfortunately, proving the pre-requisites for the existence of $\epsilon$-bisimulations for complex dynamical
models, or coming up with suitable and practically tractable bisimulation functions, is extremely difficult in practice.
In addition, establishing  $\epsilon$-bisimulation 
requires full knowledge of the system dynamics making the scheme
inapplicable where one system is an actual physical component 
with unknown mathmatical dynamics.
Bisimulation notions have hence been of limited practical use.

Instead, a more pragmatic semi-formal approach has gained prominence in industrial practice.
In this approach, the two systems  are executed on the same input sequences and a metric on finite trajectories
is used to evaluate the  closeness of  these trajectories.
The key to this methodology is the selection of a {\em good} metric, with the following properties: 
\begin{compactitem}
\item \emph{Transference.}
Closeness in the metric must translate to preserving
interesting classes  of logical and functional specifications between systems, and
\item \emph{Tractability.}
The metric should be efficiently computable.
\end{compactitem} 
In addition, there is the more informal requirement of \emph{applicability}: the metric should classify systems, that
the engineers consider close, as being close, and conversely. 

A number of metrics have been proposed recently.
The simplest is a 
{\em pointwise} metric that computes the maximum pointwise difference between two trajectories,
sometimes generalized to apply a constant time-shift 
to one trajectory~\cite{DonzeM10}.
Unfortunately, for many practical models, two trajectories may be close only under variable time-shifts. 
This is the case, for example, for two dynamical models that may use different numerical integration techniques
(e.g., fixed step versus adaptive step) or when some component in the implementation has some jitter.
Thus, the pointwise metric spuriously report large distances for ``close'' models.
More complicated hybrid distances have been proposed \cite{GeorgiosMemo14}.
The transference properties of these metrics w.r.t.\ 
common temporal logics for dynamical systems are not yet clear.

In this work we present a methodology for quantifying conformance between
real-valued dynamical systems based on the \emph{Skorokhod} 
metric~\cite{Davoren09}.
The Skorokhod metric allows for mismatches in both the trace values \emph{and} in the 
timeline, and quantifies temporal and spatial variation of the system dynamics under
a unifying framework.
The distortion of the timeline is specified by a \emph{retiming} function $\retime$ which
is a continuous bijective strictly increasing function from $\reals_+$ to  $\reals_+$.
Using the retiming function, we obtain the \emph{retimed trace} $x\left(\retime(t)\right)$ from the
original trace $x(t)$.
Intuitively,  in the retimed trace $x\left(\retime(t)\right)$,
we see exactly the same values as before, in
exactly the same order,  but the
time duration between two values might now be different than
the corresponding duration in the original trace.
The amount of distortion for the retiming $\retime$  is given by $\sup_{t\geq 0} \abs{\retime(t)-t}$.
Using retiming functions, the Skorokhod distance between two traces $x$ and $y$ is defined to be
the least value over all possible retimings $\retime$ of:
\vspace{-1mm}
\[\max\left(\sup_{t\in[0,T]} \abs{\retime(t)-t},\,  \sup_{t\in [0,T]}\dist\big(x\left(\retime(t)\right), y(t)\big)
\right),\]
\vspace{-1mm}
where $\dist$ is a pointwise metric on values.
The Skorokhod distance thus
incorporates two components: the first component quantifies
the {\em timing discrepancy} of the timing distortion required to ``match'' two traces,
and the second quantifies the  \emph{value mismatch}  (in the metric space ${\myO}$)
of the values under the timing distortion.
The Skorokhod metric was introduced as a theoretical basis for defining the semantics of hybrid systems by providing
an appropriate hybrid topology \cite{CaspiB02,Broucke98}.
We now demonstrate its usefulness in the context of conformance testing.

\smallskip
\noindent
\textbf{Transference.}
We show that the Skorokhod metric gives a robust quantification of system conformance
by relating the metric to \TLTL (timed \LTL) enriched with 
(i)~predicates of the form $f(x_1,\dots, x_n) \geq 0$, as in Signal Temporal Logic,
for specifying constraints on trace values; and
(ii) \emph{freeze quantifiers}, as in \TPTL~\cite{AlurH94},  for specifying temporal constraints
(freeze quantifiers can express more complex timing constraints than bounded
timing constraints, \emph{e.g}  of \MTL).
This logic subsumes the \MITL-based logic \STL~\cite{DonzeM10}.
We prove a \emph{transference theorem}: flows (and propositional traces) which are close under
the Skorokhod metric satisfy ``close''  \TLTL formulae for a rich class of temporal
and spatial  predicates; where the untimed structure of the formulae 
remains unchanged, only the predicates are enlarged.



\noindent\textbf{Tractability.}
We improve on recent polynomial-time algorithms for the Skorokhod 
metric~\cite{MajumdarPHSCC15}
by taking advantage of the fact that, in practice, only retimings
that map the times in one trace to  ``close'' times in the other are of interest.
This enables us to obtain a streaming 
sliding-window based monitoring procedure which takes only
$O(W)$ time per sample, where $W$ is the window size 
(assuming the dimension $n$ of the system to be a constant).

\smallskip\noindent\textbf{Usability.}
Using the Skorokhod distance checking procedure as a subroutine,
we have implemented a Simulink toolbox for conformance testing.
Our tool integrates with Simulink's model-based design flow for 
control systems, and provides a stochastic search-based approach to
find inputs which maximize the Skorokhod distance between 
systems under these inputs.

We present three case studies from the control domain, including industrial challenge problems;
our empirical evaluation 
shows that our tool computes  sharp estimates of the conformance distance reasonably
fast on each of them.
Our input models were complex enough that more theoretically appealing techniques such as
$\epsilon$-bisimulation function generation could not be applied.
In particular, we demonstrate how two models that only differ in the underlying ODE solver
can nevertheless deviate enough to invalidate system requirements on settling time.

We conclude that the Skorokhod metric can be an effective foundation for semi-formal conformance
testing for complex dynamical models.
Proofs of the theorems are given in the accompanying technical report [REF].
%

\smallskip\noindent\textbf{Related Work.}
The work of~\cite{GeorgiosMemo14,GeorgiosHFDKU14} is closely related
to ours.
In it, robustness properties of hybrid state \emph{sequences} are derived
with respect to a trace metric which also 
 quantifies temporal and spatial variations. 
Our work differs in the following ways.
First, we  guarantee  robustness properties over \emph{flows} rather than
only over (discrete) sequences.
Second, the Skorokhod metric is a stronger form of the 
$(T,J,(\tau, \epsilon))$-closeness degree\footnote{Instead of
having two separate parameters $\tau$ and $\epsilon$ for
time and state variation, we pre-scale time and the $n$ state components
with  $n+1$ constants, and have a single value quantifying closeness of the
scaled traces.}\textsuperscript{,}\footnote{Informally, two signals
$x,y$ are $(T,J,(\tau, \epsilon))$-close if for each point
$x(t)$, there is a point $y(t')$ with $|t-t'| < \tau$ such that
$\dist(x(t), y(t')) <\epsilon$; and similarly for $y(t)$.}(for systems
which do not have hybrid time); and
allows us to give stronger robustness transference guarantees.
The Skorokhod metric requires order preservation of the timeline, which
the $(T,J,(\tau, \epsilon))$-closeness function does not.
Preservation of the timeline order allows us to
(i)~keep the untimed structure of the formulae the same (unlike
in the transference theorem of~\cite{GeorgiosMemo14});
(ii)~show transference of a rich class of global timing constraints
using freeze quantifiers (rather than only for the standard bounded time
quantifiers of \MTL/\MITL).
However, for implementations where the timeline order is not preserved,
we have to settle for the less stronger guarantees provided 
by~\cite{GeorgiosMemo14}.
The work of~\cite{DonzeM10}, in terms of robustness,
 deals mainly with spatial robustness of \STL;
the only temporal disturbances considered are constant time-shifts for the entire signal
where the entire signal is moved to the past, or to the future by the same amount.
The Skorokhod metric incorporates time-shifts which are variable along the timeline.

\vspace{-4mm}
\section{Preliminaries}
\vspace{-3mm}

\noindent\textbf{Traces.}
A (finite) \emph{trace} or a  \emph{signal} 
$\pi: [T_i,T_e] \mapsto {\myO} $ is a  mapping from a
finite closed interval $[T_i,T_e]$ of $\reals_+$,
with $0 \leq T_i < T_e$,
 to some topological space  ${\myO}$. 
 If $\myO$ is a metric space, we refer to the  associated metric as $\dist_{\myO}$.
The time-domain of $\pi$, denoted $\tdom(\pi)$  is the time domain $[T_i,T_e]$ 
over which it is defined.
The  time-duration of $\pi$, denoted as $\tlen(\pi)$, is $\sup\left( \tdom(\pi) \right)$.
The $t$-suffix of  $\pi$ for $t\in \tdom(\pi)$, denoted by $\pi^t$, is the 
trace $\pi$ restricted to the interval $( \tdom(\pi) \cap [t, \tlen(\pi)]$.
We denote by $\pi_{\downarrow T'_e}$ the prefix trace obtained from
$\pi$ by restricting the domain to $[T_i,T'_e]\subseteq \tdom(\pi)$.

\noindent\textbf{Systems.}
A (continuous-time) \emph{system}  
$\system: \left( \reals_+^{\scalebox{0.6}{[\ ]}}\mapsto {\myO}_{\myinput}\right)\, \mapsto\,
 \left( \reals_+^{\scalebox{0.6}{[\ ]}}\mapsto {\myO}_{\myoutput}\right)$,
where $ \reals_+^{\scalebox{0.6}{[\ ]}}$ is the set of finite closed intervals of $\reals_+$,
 transforms input traces
$\pi_{\myinput}: [T_i,T_e] \mapsto {\myO}_{\myinput} $ into  output traces
$\pi_{\myoutput}: [T_i,T_e] \mapsto {\myO}_{\myoutput} $ (over the same time domain).
We require  that if $\system(\pi_{\myinput}) \mapsto \pi_{\myoutput}$,
then for every $\min \tdom(\pi) \leq T_e' < \max \tdom(\pi)$, the system $\system$ maps
${\pi_{\myinput}}_{\downarrow T'_e} $ to ${\pi_{\myoutput}}_{\downarrow T'_e} $.
Thus, we only consider causal systems.
Common examples of such systems are (causal) dynamical, and hybrid dynamical 
systems~\cite{Branicky1995PhD,tabuadabook}.

\noindent\textbf{Conformance.}
A system $\system'$ conforms to the system $\system$ over
an input trace $\pi_{\myinput}$ if
$\system'(\pi_{\myinput}) = \system(\pi_{\myinput}) $,
\emph{i.e.} if the behavior of $\system'$ on  the input trace $\pi_{\myinput}$ 
is the same as that of $\system$.
The system  $\system'$ conforms to the system $\system$ over
the input trace set $\Pi_{\myinput}$ if conformance holds for each input trace
in $\Pi_{\myinput}$.
Given a metric $\dist$ over input traces, and an input trace set $\Pi_{\myinput}$,
the \emph{quantitative conformance} between $\system'$ and
 $\system$ over $\Pi_{\myinput}$ is defined as the quantity
$
\sup_{\pi_{\myinput}\in \Pi_{\myinput}} 
\dist\left(\system'\left(\pi_{\myinput}\right),  
\system\left(\pi_{\myinput}\right) \right).
$
If $\Pi_{\myinput}$ is the set of all input traces, this quantity is the 
distance between the two systems.

\smallskip\noindent\textbf{Retimings.}
A \emph{retiming} $\retime: I \mapsto  I' $, for closed
intervals $I, I'$ of $\reals_+$   is an order-preserving (\emph{i.e.} monotone)  
continuous
bijective function from
$I$ to $I'$; 
thus if $t<t'$ then $\retime(t) < \retime(t')$.
Let the class of retiming functions from $I$ to $ I' $
be denoted as $\retimeclass_{I \mapsto  I' }$, and let
$\iden$ be the identity retiming.
Intuitively, retiming can be thought of as follows: imagine a stretchable and compressible
timeline; a retiming  of the original timeline gives a new timeline
 where some parts have been
stretched, and some compressed, without the timeline having been broken.
Given a trace  $\pi: I_{\pi} \rightarrow {\myO} $, and a retiming
$\retime: I \mapsto   I_{\pi} $; the function
$\pi\circ \retime$ is another trace from $I$ to ${\myO}$.

\vspace{-1mm}
\begin{definition}[Skorokhod Metric]
Given a retiming $\retime:  I \mapsto I' $, let $||\retime-\iden||_{\sup}$ be defined as
$
||\retime-\iden||_{\sup} =\sup_{t\in I }|\retime(t)-t|$.
Given two traces $\pi: I_{\pi}\mapsto {\myO} $ and $\pi': I_{\pi'} \mapsto {\myO} $,
where ${\myO}$ is a metric space with the associated
metric $ \dist_{{\myO}}$,
 and a retiming $\retime:  I_{\pi} \mapsto   I_{\pi'}$, let
$\norm{\pi\,-\, \pi'\circ \retime}_{\sup}$  be defined as:
\vspace{-2mm}
\[
\norm{\pi\,-\, \pi'\circ \retime}_{\sup}
=
\sup\nolimits_{t\in I_{\pi}} \dist_{{\myO}}\big(\,  \pi(t)\ ,\ \pi'\left(\retime(t)\right)\,  \big).\]
The \emph{Skorokhod distance}\footnote{
 The two components
of the Skorokhod distance (the retiming, and the value difference components) can
be weighed with different weights -- this simply corresponds to a change of scale.}
 between the traces  $\pi()$ and $\pi'()$  is defined to be:
\begin{equation}
\label{equation:Skoro}
\dist_{\skoro}(\pi,\pi') = \inf_{r\in  \retimeclass_{ I_{\pi} \mapsto   I_{\pi'}}}
\max(\norm{\retime-\iden}_{\sup} \, ,\,  \norm{\pi\,-\, \pi'\circ \retime}_{\sup}).\qed
\end{equation}
\end{definition}

\vspace{-2mm}
Intuitively, the  Skorokhod distance
incorporates two components: the first component quantifies
the {\em timing discrepancy} of the timing distortion required to ``match'' two traces,
and the second quantifies the  \emph{value mismatch}  (in the metric space ${\myO}$)
of the values under the timing distortion.
In the retimed trace $ \pi\circ \retime$, we see exactly the same values as in $\pi$, in
exactly the same order, but the times at which the value are seen can be different.

\smallskip\noindent\textbf{Polygonal Traces.}
A 
polygonal trace $\pi:  I_{\pi} \mapsto \myO$ where $\myO$ is a vector space
with the scalar field  $\reals$  is a continuous trace such that
there exists a finite sequence $\min I_{\pi}= t_0 < t_1 < \dots < t_m = \max I_{\pi}$ 
of time-points such that
the trace segment between $t_k$ and $t_{k+1}$ is affine for all $0\leq k < m$,
\emph{i.e.}, for $t_k \leq t \leq t_{k+1}$ we have
$\pi(t) = \pi(t_k) + \frac{t- t_k}{t_{k+1}-t_k}\vdot \left(\pi( t_{k+1}) - \pi(t_k)\right)$.
Polygonal traces are obtained when discrete-time traces are completed by
linear interpolation.
We remark that after retiming, the retimed trace $\pi \circ \retime$ \emph{need not}
be piecewise linear (see \emph{e.g.}~\cite{MajumdarP14}).

\vspace{-2mm}
\begin{theorem}[Computing the Distance between  Polygonal 
Traces~\cite{MajumdarPHSCC15}]
\label{theorem:SkoroFinal}
Let $\pi : I_{\pi} \mapsto \reals^n$  and $\pi': I_{\pi'}\mapsto \reals^n$ be two 
polygonal  traces  with $m_{\pi}$ and $m_{\pi'}$ affine segments respectively.
Let the  Skorokhod distance between them (for the $L_2$ norm on $\reals^n$) be
denoted as $\dist_{\skoro} (\pi, \pi') $.
\begin{compactenum}
\item Given  $\delta \geq 0$, it can be checked whether
 $\dist_{\skoro} (\pi, \pi') \leq \delta$ in time $O\left(m_{\pi}\vdot m_{\pi'}\vdot n\right)$.
\item 
  Suppose we restrict retimings to be such that the $i$-th affine segment of $\pi$
can only be matched to $\pi'$ affine segments $i-W$ through $i+W$ for all $i$,
where $W\geq 1$.
Under this retiming restriction, we can determine, with a streaming algorithm,
whether
 $\dist_{\skoro} (\pi, \pi') \leq \delta$ in time $O\left(\left(m_{\pi}+ m_{\pi'}\right)\vdot n\vdot W\right)$.\qed
\end{compactenum}
\end{theorem}
\vspace{-2mm}
Let us denote by $\dist_{\skoro}^W (\pi, \pi') $ the Skorokhod difference
between $\pi, \pi'$ under the retiming restriction of the second part of
Theorem~\ref{theorem:SkoroFinal}, \emph{i.e.}, the value obtained by
 restricting the retimings in Equation~\ref{equation:Skoro}\footnote{$\dist_{\skoro}^W$
is not a metric over traces (the triangle inequality fails).}.
The  value  $ \dist_{\skoro}^W (\pi, \pi') $
is an upper bound on $ \dist_{\skoro} (\pi, \pi') $.
In addition, for $W' < W$, we have $ \dist_{\skoro}^W (\pi, \pi')  \leq 
 \dist_{\skoro}^{W'} (\pi, \pi') $.

\vspace{-4mm}
\section{Skorokhod Distance based Conformance Testing}
\label{sec:conformance}
\vspace{-3mm}


In conformance testing, we test for the variance in behavior of two given
systems $\system_1$ and $\system_2$ under the same input\footnote{It is
also possible to extend our approach to allow inputs that are within some
bounded Skorokhod distance.}.  Given the same input, the two systems
produce potentially differing output traces; the goal is to quantify this
difference, and to determine an  input signal that causes the corresponding
output signals to exceed a user provided bound on the maximum tolerable
output trace distance.

\newcommand{\assign}{\leftarrow}
\newcommand{\cost}{\mathit{cost}}
\newcommand{\maxCost}{\mathit{maxCost}}
\newcommand{\maxIter}{\mathit{maxIterations}}

{\small
\begin{algorithm}
\DontPrintSemicolon
\caption{Algorithm to test if  $\displaystyle\max_{y_1 , y_2} \dist_\skoro(y_1,y_2) \, <\delta$}
\label{algo:simuskoro}
\KwIn{Systems $\system_1$,  $\system_2$, Bound $\delta$, Input Parameterization
$(P,F,B)$,  Time Horizon $T$}
\KwOut{$u(t)$ s.t. $y_1 = \system_1(u)$, $y_2 = \system_2(u)$, and $\dist_{\skoro}(y_1,y_2) > \delta$} 
$u \assign \mathtt{random}(P,F,B)$ \;
$\maxCost \assign -\infty$, $m \assign 0$ \;
\While{($maxCost < \delta$) \texttt{or} ($m < \maxIter$)}{
   $y_1 \assign \mathtt{simulate}(M_1, u, T)$ \;
   $y_2 \assign \mathtt{simulate}(M_2, u, T)$ \;
   $\cost \assign \dist_{\skoro}(y_1,y_2)$ \;
   \lIf{$\cost > \maxCost$}{
        $\maxCost  \assign \cost$ 
   }
   $u \assign \mathtt{pickNewInputs}(\cost)$ \;
   $m \assign m+1$\;
}
\end{algorithm}
}

Algorithm~\ref{algo:simuskoro} is a standard optimization-guided testing
algorithm in which we have used the Skorokhod distance between two output
traces as the cost function.  In such algorithms, it is common to define a
finite parameterization of the input space, represented by the tuple
$(P,F,B)$, where $P = \{p_1,\ldots,p_k\}$ represents a set of parameters,
$F = \{f_1,\ldots,f_k\}$ represents a finite set of basis functions from
$[0,T]$ to $\reals^n$, where $T$ is some finite time-horizon, and for each
$p_i \in P$, there is a $b_i \in B$ that is a closed interval in $\reals$
over which $p_i$ is assumed to take values. An input signal $u$ is defined
such that, for all $t$, $u(t) = \sum_i p_i\cdot f_i(t)$. A valid input
signal has the property that for all $i$, $p_i \in b_i$.

In each step, the algorithm picks an input signal $u$ and computes the
Skorokhod distance between the corresponding outputs $y_1 = \system_1(u)$
and $y_2 = \system_2(u)$.  Based on heuristics that rely on the current
cost, and a possibly bounded history of costs, the procedure then picks a
new value for $u$.  For instance, in a gradient-ascent based procedure, the
new value of $u$ is chosen by estimating the local gradient in each
direction in the input-parameter space, and then picking the direction that
has the largest (positive) gradient. In our implementation, we use the
Nelder-Mead (or nonlinear simplex) algorithm.

The algorithm terminates when a violation is found (i.e., a pair of inputs
that exceed the user-provided Skorokhod distance bound), or when the number
of iterations is exhausted.  The  Skorokhod distance bound $\delta$ is
chosen based on engineering requirements, \emph{e.g.}, based on the maximum
allowed weakening of the temporal logical properties that have been
verified/tested on one system.

\smallskip\noindent\textbf{Sampling and Polygonal Approximations.}
In practice, the output behaviors of the systems  are observed with
a sampling process, thus in implementations of Algorithm~\ref{algo:simuskoro},
entities $y_1$ and $y_2$ on lines $4,5$ are time-sampled output trace
\emph{sequences}, from which
 the Skorokhod distance algorithm of Theorem~\ref{theorem:SkoroFinal}
constructs (continuous time)  signals using linear
interpolation.
Given a timed trace sequence $\tseq$, let $\symb{\tseq}_{\LI}$ denote the
continuous time trace obtained from $\tseq$ by linear interpolation.
Let $\tseq_{\pi}, \tseq_{\pi'}$ be two corresponding samplings
of the traces  $\pi, \pi'$.
Since the Skorokhod distance is a metric, we have that
\vspace{-1mm}
\[\dist_{\skoro}(\pi, \pi') \leq 
\dist_{\skoro}\left(\symb{\tseq_{\pi}}_{\LI}, \symb{\tseq_{\pi'}}_{\LI}\right)\, +\, 
\dist_{\skoro}\left(\symb{\tseq_{\pi}}_{\LI}, \pi\right) + 
\dist_{\skoro}\left(\symb{\tseq_{\pi'}}_{\LI}, \pi'\right). \vspace{-1mm}\]
If $\Delta_{\sampleerr}$ is a bound on the distance between a trace, and an
interpolated completion of its sampling, we have that
$
\dist_{\skoro}(\pi, \pi') \leq \dist_{\skoro}\left(\symb{\tseq_{\pi}}_{\LI}, \symb{\tseq_{\pi'}}_{\LI}\right)\, +\, 
2\vdot \Delta_{\sampleerr}$.
Thus, in a sampling framework, a value of $2\vdot \Delta_{\sampleerr}$
needs to be added to the Skorokhod distance between the polygonal approximations.

Section~\ref{section:Logic} presents a theory of (quantifiable) 
transference of logical  properties.
Section~\ref{section:Experiment} presents results on our implementation
of Algorithm~\ref{algo:simuskoro}.
We also discuss several case studies, 
providing rationale for choosing the appropriate $\delta$ value, 
and present results on the
computation time and the conformance distance found.

\vspace{-2mm}
\section{Transference of  Logical Properties}
\label{section:Logic}

\vspace{-2mm}
In this section, we demonstrate transference of logical properties.
If two traces are at a distance of $\delta$, and one trace satisfies a logical
specification $\phi$, we derive the ``relaxation'' needed (if any) in $\phi$ so that
the other trace also satisfies this relaxed logical specification.
The logic we use  is a version of  the  timed linear time logic 
\TLTL~\cite{AlurH94} (a timed version of the logic \LTL).
We show  that the Skorokhod distance provides robust transference of
specifications in this logic: if the Skorokhod distance between
two traces is small, they satisfy close  \TLTL formulae.
We first present the results in a propositional framework, and then extend
to $\reals^n$-valued spaces.

\vspace{-3mm}
\subsection{The Logic \TLTL}
\vspace{-1mm}

Let $\prop$ be a set of propositions.
A \emph{propositional trace} $\pi$ over $\prop$ is a 
trace where the topological space is $2^{\prop}$, with the associated metric:
$\dist_{\prop}(\sigma, \sigma')  = \infty$ if $\sigma \neq \sigma'$, and $0$ otherwise
for $\sigma,\sigma'\in 2^{\prop}$.
We restrict our attention to propositional 
traces with finite variability: we require
that there exists a finite partition of $\tdom(\pi)$ into disjoint subintervals
$I_0, I_1, \dots, I_m$ such that $\pi$ is constant on each subinterval.
%
The set of all  timed propositional traces over $\prop$  is denoted by $\Pi(\prop)$.

\vspace{-1mm}
\begin{definition}[\mTLTL{\small$(\TFun)$} Syntax]
Given a set of propositions $\prop$, a set of (time) variables $V_{\mytime}$, and a set $\TFun$ of functions
from $\reals_+^l$ to $\reals$,
the formulae of  \mTLTL($\TFun$)
are defined by the following grammar.
\[
\phi := p \mid \true \mid f_{\mytime}(\overline{x}) \sim 0 \mid \neg\phi \mid \phi_1 \wedge \phi_2 \mid  \phi_1 \vee \phi_2 \mid \phi_1 \until \phi_2 \mid x.\phi
\quad \text{ where}
\]
\begin{compactitem}
\item 
$p\in \prop$ and  $x\in V_{\mytime}$, and $\overline{x} = (x_1, \dots, x_l)$ with $x_i\in V_{\mytime}$ for all $1\leq i \leq l$;
\item  $f_{\mytime} \in \TFun$ is a   real-valued function,
and $\sim$ is one of $ \set{\leq, <, \geq, >}$.\qed
\end{compactitem}

\end{definition}

We say that the variable $x$ is \emph{bound} in $\phi$ if $\phi$ is $ x.\Psi$, otherwise it is \emph{free}.
The quantifier ``$x.$'' is known as the \emph{freeze quantifier}, and binds the variable $x$ to the current
time.
A formula is \emph{closed} if it has no free variables.

\begin{definition}[\mTLTL{\small$(\TFun)$} Semantics]
\label{definition:PropositionalSemantics}
Let $\pi: I \mapsto 2^{\prop}$ be a timed propositional trace,
$t_0 = \min(I)$, 
and let $\env: V \mapsto I$ be the time environment
mapping the variables in  $V$ to time values in $I$.
The satisfaction of the timed sequence $\pi$ with respect to the  \mTLTL($\TFun)$ formula $\phi$ in the time environment $\env$ is
written as $\pi \models_{\env} \phi$, and  is
defined inductively as follows (denoting $t_0 = \min \tdom(\pi)$).
\begin{align*}
\pi & \models_{\env} p \text{ for } p\in \prop  \text{ iff } p\in \pi(t_0);
\qquad \pi  \models_{\env} \true;
\qquad \pi  \models_{\env} \neg\Psi \text{ iff } \pi  \not\models_{\env} \Psi;\\
\pi & \models_{\env} \phi_1 \wedge \phi_2 \text{ iff } \pi  \models_{\env} \phi_1 \text{ and } \pi \models_{\env} \phi_2;
\quad \pi  \models_{\env} \phi_1 \vee \phi_2 \text{ iff } \pi  \models_{\env} \phi_1 \text{ or } \pi \models_{\env} \phi_2;\\
 \pi & \models_{\env} f_{\mytime}(x_1, \dots, x_l) \sim 0 \text{ iff }  f_{\mytime}(\env(x_1), \dots, \env(x_l)) \sim 0
\text{ for }  \sim \,\in \set{\leq, <, \geq, >};\\
\pi & \models_{\env}\! x.\psi \text{ iff } \pi\!  \models_{\env[x\!:=\!t_0]}\! \psi \text{ where } \env[x\!:=\!t_0] \text{ agrees}\, \text{with }
\env\! \text{ for}\, \text{all } z\!\neq\! x, \!\text{ and}\,\text{maps}\, x\, \text{to } t_0; \\
\pi & \models_{\env}  \phi_1 \until \phi_2  \text{ iff  }  \pi^t  \models_{\env} 
\phi_2 \text{ for some }  t\in I
\text{ and }  \pi^{t'}  \models_{\env} \phi_1\vee \phi_2  \text{ for all }  t_0\leq t' < t.
\end{align*}
A timed trace $\pi$ is said to satisfy the closed formula $\phi$ (written
as  $ \pi \models \phi$) if there is some environment $\env$ such that
$ \pi \models_{\env} \phi$.
\qed
\end{definition}
\vspace{-1mm}
The definition of additional temporal operators in terms of these base operators  is standard:
the ``eventually'' operator $\Diamond \phi$ stands for $\true \until \phi$; and the
``always'' operator $\Box \phi$ stands for  $\neg \Diamond \neg \phi$.
 \mTLTL{\small$(\TFun)$} provides a richer framework than \MTL~\cite{Koymans90} 
for expressing timing
constraints as:
(i)~freeze quantifiers
allow specification of constraints between distant contexts, which the bounded temporal operators in \MTL
cannot do; and
(ii)~the predicates $f_{\mytime}() \sim 0$ for  $f_{\mytime}\in \TFun$ allow the
specification of  complex timing requirements not expressible in \MTL.
\vspace{-3mm}
\begin{example}[Freeze quantifiers; \mTLTL{\small$(\TFun)$} subsumes \MTL]
Let $\TFun$ be the set of two variable functions  of the form $f(x,y) = x-y+c$ where $c$ is a rational  constant.
Then \mTLTL{\small$(\TFun)$} subsumes \MTL.
The \MTL formula $Q \until_{[a,b]} R$ can be written as
\vspace{-1mm}
\[
x. \Big(Q \until y. \big( \,(y\leq x+b) \wedge (y\geq x+a)  \wedge  R\big)\Big).\]
\vspace{-1mm}
We explain the formula as follows.
We assign the ``current'' time $t_x$  to the variable $x$, and some future time $t_y$ to the variable $y$.
The values $t_x$ and  $t_y$ are such that
at time $t_y$, we have  $R$ to be true, and moreover, at all times between $t_x$  and $t_y$, 
we  have $Q\vee R$ to be true.
Furthermore, $t_y$ must be such that $t_y\in [t_x+a, t_x+b]$, which is specified by the term
$(y\leq x+b) \wedge (y\geq x+a) $.
\qed
\end{example}
\vspace{-4mm}
\begin{example}[Temporal Constraints]
\label{example:TemporalConstraints}
Suppose we want to express that  whenever the event $Q$ occurs, 
it must be followed by a response $R$, and then
by $S$.
In addition, we have the following timing requirement:
if $\varepsilon_{QR}, \varepsilon_{RS}, \varepsilon_{QS}$ are
 the time delays  between $Q$ and $R$, between $R,S$, and between
$Q$ and $S$ respectively, then:
we must have 
$\varepsilon^2_{QR}+ \varepsilon^2_{RS}+\varepsilon^2_{QS} \leq d$ for
a given positive constant $d$.
This can be written using freeze quantifiers as the \TLTL formula $\phi$:
\[
x. \left( Q \rightarrow \Diamond 
\big(y. \left(R \wedge 
 \Diamond\left[z. \left(S \wedge  
\left(  (y-x)^2 + (z-y)^2 + (z-x)^2 \leq d \right) \right)
\right] \right)
\big)\right).\qed\]
\end{example}

\vspace{-3mm}
\subsection{Transference of \TLTL Properties for Propositional Traces}
\vspace{-1mm}

We show in this section that if a timed propositional trace
 $\pi$ satisfies a \mTLTL{\small$(\TFun)$} formula $\phi$, then any timed trace $\pi'$ that is at most
$\delta$ distance away from $\pi$ satisfies  a slightly relaxed version of the formula $\phi$, 
the degree of relaxation being governed by $\delta$; and
the variance of the  functions in $\TFun$ over the time interval
containing the time domains of $\pi$ and $\pi'$.

Recall that the distance between  two sets of propositions $\sigma, \sigma'$ is
$\infty$ if $\sigma\neq \sigma'$, and $0$ if  $\sigma = \sigma'$.
The distance between two  propositional traces is defined
to be the Skorokhod distance with the aforementioned metric on $2^{\prop}$.



Next, we define relaxations of \mTLTL{\small$(\TFun)$}formulae.
The relaxations are defined as a syntactic transformation 
on formulae which do not have negations, except on the propositions.
Every  \mTLTL{\small$(\TFun)$}formula can be expressed in this negation-normal form.
To remove negations from the until operator, we use the waiting for operator, $\awaits$, defined as:
\begin{verse}
\vspace{-2mm}
$\pi\models_{\env} \phi_1 \awaits \phi_2$ iff either (1)~$\pi^t \models_{\env}\phi_1$ for all
 $ t\in I $;
or (2)~$\pi^t  \models_{\env}  \phi_2$  for some   $ t\in I $;
and $\pi^{t'}  \models_{\env} \phi_1\vee \phi_2$  for all  $t_0\leq t' < t$.
\vspace{-2mm}
\end{verse}
It can be showed that every  \mTLTL{\small$(\TFun)$} formula can be rewritten
using the $\awaits$ operator such that negations appear only over the propositions
(the procedure is given in the Appendix).

\vspace{-1mm}
\begin{definition}[$\delta$-relaxation of \mTLTL{\small$(\TFun)$} formulae]
\label{definition:Relaxation}
Let $\phi$ be a  \mTLTL{\small$(\TFun)$} formula in which negations appear only on the 
propositional symbols.
The $\delta$ relaxation of $\phi$ (for $\delta\! \geq\! 0$) over  a closed interval $J$,
denoted 
$\myrelax_{J}^{\delta}(\phi)$, is defined as:
\vspace{-1mm}
%
\[
\begin{array}{l|l}
\begin{array}{lll}
\myrelax_{ J}^{ \delta}(p)     & = & p      \\
\myrelax_{ J}^{\delta}(\neg p) & = & \neg p \\
\myrelax_{ J}^{\delta}( \phi_1 \wedge \phi_2 ) & = &
\myrelax_{ J}^{\delta}( \phi_1)  \wedge \myrelax_{J}^{\delta}(\phi_2 ) \\
\myrelax_{ J}^{\delta}( x.\psi ) & = & x.\myrelax_{ J}^{\delta}(\psi)\\
\myrelax_{ J}^{\delta}(  \phi_1 \until \phi_2 ) & = &
\myrelax_{ J}^{\delta}(\phi_1) \until \myrelax_{J}^{\delta}(\phi_2) \\
\end{array} &
\begin{array}{lll}
\myrelax_{ J}^{\delta}(\true)  & = & \true \\
\myrelax_{ J}^{\delta}(\false) & = & \false \\
\myrelax_{ J}^{\delta}( \phi_1 \vee \phi_2 ) & = &
\myrelax_{ J}^{\delta}( \phi_1)  \vee \myrelax_{ J}^{\delta}(\phi_2 )\\
 & & \\
\myrelax_{ J}^{\delta}(  \phi_1 \awaits \phi_2 ) & = &
\myrelax_{ J}^{\delta}(\phi_1) \awaits \myrelax_{ J}^{\delta}(\phi_2) \\
\end{array}
\end{array}
\]
%
%
%
\newcommand{\MyDef}{\mathrm{def}}
\newcommand{\eqdefU}{\ensuremath{\mathop{\overset{\MyDef}{=}}}}
\newcommand{\eqdef}{\mathop{\overset{\MyDef}{\resizebox{\widthof{\eqdefU}}{\heightof{=}}{=}}}}
\vspace{-0.4em}
\begin{equation}
\label{equation:RelaxProp}
\begin{array}{l}
\vspace{0.3em}
\myrelax_{ J}^{\delta}\left(f_{\mytime}(x_1, \dots, x_l)  \right) \sim 0) = 
\left\{
        \begin{array}{ll}
        f_{\mytime}(x_1, \dots, x_l) \,  +\,  K_{f_{\mytime}}^I(\delta) \  \sim\,  0 & \text{ if } 
         \sim\, \in\set{>, \geq}\\
        f_{\mytime}(x_1, \dots, x_l) \, -\,   K_{f_{\mytime}}^I(\delta) \  \sim \, 0 & \text{ if } 
         \sim \,\in\set{<, \leq},\\
        \end{array}
\right. \\
\vspace{0.3em}
\text{where } K_{f_{\mytime}}^I:[0, \max \tdom( J)\, - \, \min \tdom(
J)]\mapsto \reals_+ \text{, and}  \\
K_{f_{\mytime}}^I(\delta) \eqdef
\displaystyle\sup_{
  \begin{array}{l}
    {t_1,\dots,t_l\in  J}\\
    {t_1',\dots,t_l'\in  J}
  \end{array}
} \left\{
\left\arrowvert
  \begin{array}{c}
    f_{\mytime}(t_1,\dots, t_l) \\
    -\\
    f_{\mytime}(t_1',\dots,t_l')
  \end{array}\right\arrowvert \ 
\text{ s.t. }
  |t_i-t_i'| \leq \delta \text{ for all } i
\right\} 
\end{array}
\hspace{-3em}
\end{equation}
\end{definition}
\vspace{-2mm}
Thus, instead of comparing the $f_{\mytime}()$ values to $0$, we relax by comparing
instead to $\pm K_{f_{\mytime}}^J(\delta)$.
The other cases recursively relax the subformulae.
The functions  $K_{f_{\mytime}}^J(\delta)$ define the maximal change in the value
of $f_{\mytime}$ that can occur when the input variables can vary by $\delta$.
The role of $J$ is the above definition is to restrict the domain of the freeze
quantifier variables to the time interval $J$ (from $\reals_+$) 
 in order to obtain the least possible relaxation on a given trace
$\pi$ (\emph{e.g.} we do not care about the values of a function in ${\TFun}$ outside
of the domain $\tdom(\pi)$ of the trace).

\vspace{-1mm}
\begin{example}[$\delta$-relaxation for Bounded Temporal Operators -- \MTL]
 We demonstrate how 
$\delta$-relaxation operates on bounded time constraints
through an example.
Consider an \MTL formula $\phi= Q \until_{[a,b]} R$.
This can be written as a  \TLTL formula, and relaxed using the
$ \myrelax_{\reals_+}^{\delta}$ function.
The relaxed  \TLTL formula is again equivalent to an \MTL
formula, namely $Q \until_{[a-2\cdot\delta\,,\, b+2\cdot \delta]} R$.
The details are explained in Example~\ref{example:RelaxationMTL}
in the Appendix.
\qed
\end{example}

\begin{theorem}[Transference for Propositional Traces]
\label{theorem:PropositonalRobustness}
Let $\pi, \pi'$ be two  timed propositional traces such that 
$\dist(\pi, \pi') < \delta$ for some finite $\delta$.
Let 
$\phi$ be a closed  \mTLTL{\small$(\TFun)$} formula in negation-normal form.
If $\pi\models \phi$, then
$\pi'\models \myrelax_{I_{\pi,\pi'}}^{\delta}(\phi)$ 
where
${I_{\pi,\pi'}}$ is the convex hull of $\tdom(\pi) \cup \tdom(\pi')$.
\qed
\end{theorem}
Theorem~\ref{theorem:PropositonalRobustness} relaxes the freeze variables
over the entire signal time-range ${I_{\pi,\pi'}}$;
it can be strengthened by relaxing over a smaller range:
if $\pi\models \phi$, and $t_1, \dots, t_k$ are time-stamp assignments to the
freeze variables $x_1, \dots, x_k$  which witness $\pi$ satisfying $\phi$,
then $x_i$ only needs to be relaxed over $[t_i-\delta, t_i+\delta]$ rather
than  the larger interval ${I_{\pi,\pi'}}$.
These smaller relaxation intervals for the freeze variables can be incorporated
in Equation~\ref{equation:RelaxProp}.
We omit the details for ease of presentation.

\begin{example}
\label{example:Transference-one}
Recall Example~\ref{example:TemporalConstraints}, and the
formula $\phi$ presented in it.
Suppose a flow $\pi$ satisfies $\phi$; and let $\pi'$ be $\delta$ close to $\pi$
under the Skorokhod metric (for propositional traces).
Our robustness theorem ensures that
(i)~$\pi'$ will satisfy the same untimed formula
$Q\rightarrow \, \Diamond \left(R \wedge \Diamond S\right) $; and
(ii) it gives a bound on how much the timing constraints need to be relaxed in $\phi$
in order
to ensure satisfaction by $\pi'$;  it states that $\pi'$ satisfies the following relaxed
formula $\phi'$ for every $\epsilon >0$:
\[\pi'\models
x. \left( Q \rightarrow \Diamond
\big(y. \left(R \wedge
 \Diamond\left[z. \left(S \wedge
\left(  (y-x)^2 + (z-y)^2 + (z-x)^2 \leq d^{\dagger} \right) \right)
\right] \right)
\big)\right)\]
where 
$d^{\dagger} = d + 12 \cdot(\delta+\epsilon)^2 + 4\sqrt{3}\cdot(\delta+\epsilon)
\cdot \sqrt{d}$.
The constant $d^{\dagger} $ is derived in the appendix.
\qed
\end{example}

\vspace{-5mm}
\subsection{Transference of \TLTL properties for $\reals^n$-valued Signals}
\vspace{-1mm}
A \emph{timed $\reals^n$-valued trace} 
$\pi$ is a function from a closed interval $I$ of $\reals_+$ to $\reals^n$.
For $ \ol{\alpha} = (\alpha^0, \dots, \alpha^n)\in \reals^n$, we denote the
$k$-th dimensional value $\alpha^k$ as  $\ol{\alpha}[k]$.
The $\pi$ projected function onto the $k$-th  $\reals$ dimension   is
denoted by $\pi_{k}: I \mapsto \reals$.

In order to define the satisfaction of  \TLTL formulae  over  timed $\reals^n$-valued 
sequences,
we use booleanizing predicates $\mu: \reals^n \mapsto \bool$, as in 
\STL~\cite{DonzeM10},   to transform
$\reals^n$-valued sequences in to timed propositional sequences.
These predicates are part of the logical specification.
In this work, we restrict our attention to traces and predicates such that each predicate
varies only finitely often on the finite time traces under consideration.

\vspace{-1mm}
\begin{definition}[\mTLTL{\small$(\TFun, \SFun)$} Syntax]
Given 
 a set of variables $V_{\mytime}$ (the freeze
variables), a set of \emph{ordered} variables
 $V_{\mysig}$ (the signal variables),
and two sets $\TFun, \SFun$ of
 functions,
the formulae of  \mTLTL{\small$(\TFun, \SFun)$}
are defined by the grammar:
\[
\phi := \true \mid f_{\mytime}(\overline{x}) \sim 0  \mid 
f_{\mysig}(\overline{y}) \sim 0  \mid 
\neg\phi \mid \phi_1 \wedge \phi_2 \mid \phi_1\vee \phi_2\mid \phi_1 \until \phi_2 \mid x.\phi 
\quad \text{ where}
\]
\begin{compactitem}
\item   $x\in V_{\mytime}$, and $\overline{x} = (x_1, \dots, x_l)$ with $x_i\in V_{\mytime}$ for all $1\leq i\leq l$;
\item   $\overline{y} = (y_1, \dots, y_d)$ with $y_j\in V_{\mysig}$ for all $1\leq j \leq d$;
\item $V_{\mytime} $ and  $V_{\mysig}$ are disjoint;
\item  $f_{\mytime} \in \TFun$ and  $f_{\mysig} \in \SFun$ are  real-valued functions,
and $\sim $ is $ \leq, <, \geq, $ or $>$.\qed
\end{compactitem}

\end{definition}

\vspace{-1mm}
The semantics of \mTLTL{\small$(\TFun, \SFun)$} is straightforward and similar to the propositional case (Definition~\ref{definition:PropositionalSemantics}).
The only new ingredients are the booleanizing predicates
$f_{\mysig}(\overline{y}) \sim 0 $: we define
$\pi \models_{\env} f_{\mysig}(y_1, \dots, y_d) \sim 0$
iff $ f_{\mysig}( \pi_{j_1}[t_0], \dots, \pi_{j_d}[t_0]) \sim 0$ for any freeze variable 
environment $\env$, where $t_0 = \min \tdom(\pi)$, and $y_i$ is the $j_i$-th
variable in $V_{\mysig}$ (\emph{i.e.}, $y_i$ refers to the $j_i$-th dimension
in the signal trace).
We require that for a timed $\reals^n$-valued trace
$\pi$ to satisfy $\phi$, the arity of the functions in $ \SFun$
occurring in $\phi$ should not be more than $n$, that is, functions should not refer
to dimensions greater than $n$ for an $\reals^n$ trace.

\smallskip\noindent\textbf{$\delta$ relaxation of  \mTLTL{\small$(\TFun, \SFun)$}.}
Let $ \jmap$ be a mapping from $V_{\mysig}$  to closed intervals
of $\reals$, thus  $\jmap(z)$ denotes a sub-domain of $z\in V_{\mysig}$.
The relaxation function $\myrelax_{{J}, \jmap}^{\delta}$ which
operates on  \mTLTL{\small$(\TFun, \SFun)$} formulae is defined
analogous to the relaxation function $\myrelax_{J}^{\delta}$
in Definition~\ref{definition:Relaxation}.
We omit the similar cases, and only present the new case for the predicates
formed from $ \SFun$ (the full definition can be found in the appendix).
\vspace{-3mm}
\[
\myrelax_{J, \jmap}^{\delta}\left(f_{\mysig}(z_1, \dots, z_l)  \right) \sim 0)  
\  = 
\begin{cases}
f_{\mysig}(z_1, \dots, z_l) \,  +\,  K_{f_{\mysig}}(\delta) \  \sim\,  0 & \text{ if } 
 \sim\, \in\set{>, \geq};\\
f_{\mysig}(z_1, \dots, z_l) \, -\,   K_{f_{\mysig}}(\delta) \  \sim \, 0 & \text{ if } 
 \sim \,\in\set{<, \leq}
\end{cases}
\]
where
$K_{f_{\mysig}}: \big[0,\  \max_{z\in V_{\mysig}} |\max \jmap(z) \, - \, \min\jmap(z)| \big]
\mapsto \reals_+$
is a function s.t.
\vspace{-1mm}
\[
 K_{f_{\mysig}}(\delta) = 
\sup_{
  \begin{array}{c}
    z_i\in \jmap(z_i); \,  z'_i\in \jmap(z'_i)\\
    \text{ for all } i 
 \end{array}
}
\left\{
\left\arrowvert
  \begin{array}{c}
    f_{\mysig}(z_1,\dots, z_l) \\
    -\\
    f_{\mysig}(z_1',\dots,z_l') 
  \end{array}\right\arrowvert
\text{ s.t. }
  |z_i-z_i'| \leq \delta \text{ for all } i
\right\}.
\]
The functions  $K_{f_{\mysig}}(\delta)$ define the maximal change in the value
of $f_{\mysig}$ that can occur when the input variables can vary by $\delta$ over the
intervals in $\jmap(z)$ and $J$.
The role of $\jmap$ in the above definition is to restrict the domain of the signal variables
in order to obtain the least possible  relaxation bounds on the signal
constraints; as was done in Definition~\ref{definition:Relaxation} for the
freeze variables.

\begin{theorem}[Transference for $\reals^n$-valued  Traces]
\label{theorem:SignalTLTLRobustness}
Let $\pi, \pi'$ be two $\reals^n$-valued traces such 
the Skorokhod distance between them is less than $\delta$  for some finite $\delta$.
Let $\phi$ be a closed  \mTLTL{\small$(\TFun, \SFun)$} formula in negation-normal form.
If $\pi \models \phi$, then $\pi'\models \myrelax_{I_{\pi, \pi'}, \imap}^{\delta}(\phi)$, where
\begin{compactitem}
\item ${I_{\pi,\pi'}}$ is the convex hull of $\tdom(\pi) \cup \tdom(\pi')$; and
\item $\imap(z)$
is the convex hull of
$\set{\pi(t)[k] \mid t\in \tdom(\pi)} \cup \set{\pi'(t)[k] \mid t\in \tdom(\pi')} $;
  where $z$ is the $k$-th variable in the ordered
set $V_{\mysig}$.\qed
\end{compactitem}
\end{theorem}
\vspace{-2mm}
Theorem~\ref{theorem:SignalTLTLRobustness} can be strengthened similar to
the strengthening mentioned for Theorem~\ref{theorem:PropositonalRobustness}
by relaxing the variables over smaller intervals obtained from
assignments to variables which witness $\pi\models \phi$.

\vspace{-1mm}
\begin{example}[Spatial Constraints and Transference]
Recall Example~\ref{example:TemporalConstraints}, 
suppose that the events $Q,R,S$ are defined by the 
following predicates over real variables $\alpha_1$ and $\alpha_2$.
Let $Q \equiv \alpha_1 + 10\vdot\alpha_2 \geq 3$; the
predicate $R \equiv |\alpha_1| + |\alpha_2| \leq 20$; and
$S \equiv |\alpha_1| + |\alpha_2| \leq 15$.
Let $\pi$ satisfy this formula with these predicates, and let 
$\pi'$ be $\delta$ close to $\pi$,
for a finite $\delta$ under the Skorokhod metric for $\reals^2$.
Our robustness theorem ensures that
$\pi'$ will satisfy the relaxed formula 
\[
x. \left( Q^{\delta} \rightarrow \Diamond 
\big(y. \left(R^{\delta} \wedge 
 \Diamond\left[z. \left(S^{\delta} \wedge  
\left(  (y-x)^2 + (z-y)^2 + (z-x)^2 \leq d+12\vdot \delta^2 \right) \right)
\right] \right)
\big)\right). \]
where the relaxed predicates  $Q^{\delta},R^{\delta},S^{\delta}$
are defined as follows:
$Q^{\delta}  \equiv \alpha_1 + 10\vdot\alpha_2 \geq 3- 22\vdot\delta$;
and $R^{\delta} \equiv |\alpha_1| + |\alpha_2| \leq 20+4\vdot\delta$;
and $S^{\delta} \equiv |\alpha_1| + |\alpha_2| \leq 15+4\vdot\delta$.
\qed
\end{example}

\newcommand{\ie}{i.e.\xspace}
\section{Experimental Evaluation}
\label{section:Experiment}

\mypara{Skorokhod Distance Computation Benchmark} The Skorokhod
distance is computed with the help of a streaming, sliding window
monitoring routine which checks for a fixed $\delta$ whether the
linear interpolations of two time-sampled traces are at most $\delta$
away from each other.  The least such $\delta$ value is computed by
binary search over the monitoring routine.  The upper limit of the
search range is set to  the pointwise metric (\emph{i.e} assuming the
identity retiming) between the two traces.  The traces to the
Skorokhod procedure are pre-scaled, each dimension (and the
time-stamp) is scaled by a different constant.  The constants are
chosen so that after scaling, one unit of deviation in one dimension
is as undesirable as one unit of jitter in other dimensions.  We next
present a benchmark on the distance computing routine.

Consider the hybrid dynamical system $\system_1$ shown in
Fig.~\ref{fig:watertank}.  The system consists of two water tanks,
each with an outlet from which water drains at a constant rate $d_j$.
Both tanks share a single inlet pipe that is switched between the
tanks, filling only one tank at any given time at a constant inflow
rate of $i$. When the water-level in tank $j$ falls below level
$\ell_j$, the pipe switches to fill it.  The drain and inflow rates
$d_1$, $d_2$ and $i$ are assumed to be inputs to the system. Now
consider a version $\system_2$ that incorporates an actuation delay
that is a function of the inflow rate. This means that after the level
drops to $\ell_j$ for tank $j$, the inlet pipe starts filling it only
after a finite time.  $\system_1$ and $\system_2$ have the same
initial water level.  We perform a fixed number of simulations by
systematically choosing drain and inflow rates $d_1$, $d_2$, $i$ to
generate traces (water-level vs. time) of both systems and compute
their Skorokhod distance.  We summarize the results in
Table~\ref{tab:distance_evaluation}.

\vspace{-1mm}

\begin{figure}[t]
\centering
\begin{tikzpicture}
\tikzstyle{smalltext}=[font=\fontsize{8}{8}\selectfont]
\tikzstyle{state}=[draw,rectangle,rounded corners,minimum width=4.5em,smalltext]

\node[state] (s1) {%
$\left[\begin{array}{l}
       \dot{h_1} \\
       \dot{h_2} \\
       \end{array}\right]\!\! = \!\! 
 \left[\begin{array}{l}
       i - d_1 \\
       -d_2 \\
       \end{array}\right]$};

\node[state,right of=s1,node distance=40mm] (s2) {%
$\left[\begin{array}{l}
       \dot{h_1} \\
       \dot{h_2} \\
       \end{array}\right]\!\! = \!\! 
 \left[\begin{array}{l}
       - d_1 \\
       i -d_2 \\
       \end{array}\right]$};

\draw[->,>=latex'] (s1) to[out=30,in=150] node[smalltext,above] {$h_2 < \ell_2$} (s2); 
\draw[->,>=latex'] (s2) to[out=210,in=330] node[smalltext,below] {$h_1 < \ell_1$} (s1); 

\end{tikzpicture}
\vspace{-2mm}
\caption{System $\system_1$ used for benchmarking Skorokhod Distance computation.
Inflow rate $i$, Drain rate $d_1$ for tank $1$ and $d_2$ for tank $2$ are
all inputs to the system.\label{fig:watertank}}
\vspace{-2mm}
\end{figure}
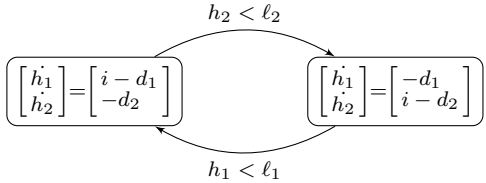

\newcommand{\muc}[1]{\multicolumn{2}{c}{#1}}

\begin{table}[t]
\caption{Benchmarking the computation of $\dist_{\skoro}(\pi_1,\pi_2)$, where
$\pi_1$ is a trace of system $\system_1$ described in Fig.~\ref{fig:watertank}, and
$\pi_2$ is a trace of system $\system_2$, which is $\system_1$ with an actuation delay. 
$\dist_2$ is the naive pointwise distance.
Both
$\pi_1$ and $\pi_2$ contain equally spaced $2001$ time points over a simulation
horizon of $100$ seconds.  \label{tab:distance_evaluation}}
\centering
\begin{tabular*}{0.99\textwidth}{@{\extracolsep{\fill}}lrrrr}
\toprule
Window size & Avg. $\dist_{\skoro}$ & \muc{Avg. Time taken (secs)} & $\max\frac{\dist_2 - \dist_\skoro}{\dist_2}$  \\
\cline{3-4}
            &                       & Computation & Monitoring &                 \\
\midrule
20          & 8.58                  &   0.81      &  0.13      &   0.09             \\
40          & 8.35                  &   1.55      &  0.26      &   0.18          \\
60          & 8.09                  &   2.31      &  0.39      &   0.26          \\
80          & 7.88                  &   3.05      &  0.52      &   0.33          \\
100         & 7.72                  &   3.77      &  0.64      &   0.38          \\
\midrule
\bottomrule
\end{tabular*}
\vspace{2mm}
\end{table}

Recall that $\dist_\skoro$ (the Skorokhod distance) computation involves a
sequence of monitoring calls with different $\delta$ values picked by a
bisection-search procedure. Thus, the total time to compute $\dist_\skoro$
is the sum over the computation times for individual monitoring calls plus
some bookkeeping. In Table~\ref{tab:distance_evaluation}, we make a
distinction between the average time to monitor traces (given a $\delta$
value), and the average time to compute $\dist_\skoro$.  There are an
average of $6$ monitoring calls per $\dist_\skoro$ computation.  We ran
$64$ simulations by choosing different input values, and then computing
$\dist_\skoro$ for increasing window sizes.  As the window size increases,
the average $\dist_\skoro$ is seen to decrease; this is expected as a
better match may be achieved in a larger window. The computation time is
also seen to increase linearly, as postulated by
Theorem~\ref{theorem:SkoroFinal}.  Finally, we see that the Skorokhod
distance is less aggressive at classifying traces as distant (as shown by
its lower overall numbers) than a simpler metric $\dist_2$ (defined as as
the maximum of the pointwise $L_2$ norm\footnote{Even though the difference
is only $38\%$ with respect to the pointwise metric, this difference is
amplified in the original state value domain, as in the experiment, the
input state values to the Skorokhod routine were scaled by $0.1$.}). We can
see this discrepancy becomes more prominent with increased window size
(because of better matches being available). 


\mypara{Case Study: LQR-based Controller} The first case study is an
example of an aircraft pitch control application taken from the openly
accessible control tutorials for Matlab and Simulink \cite{ctms}. The
authors describe a linear dynamical system of the form:
$\dot{\mathbf{x}} = (A-BK)\mathbf{x} + B\theta_{des}$.  Here,
$\mathbf{x}$ describes the vector of continuous state variables and
$\theta_{des}$ is the desired reference provided as an external input.
One of the states in the $\mathbf{x}$ vector is the pitch angle
$\theta$, which is also the system output.  The controller gain matrix
$K$ is computed using the linear quadratic regulator method
\cite{lqr}, a standard technique from optimal control.  We are
interested in studying a digital implementation of the continuous-time
controller obtained using the LQR method. To do so, we consider
sampled-data control where the controller samples the plant output,
computes, and provides the control input to the plant every $\Delta$
seconds. To model sensor delay, we add a fixed delay element to the
system; thus, the overall system now represents a delay-differential
equation.

Control engineers are typically interested in the step response of a
system. In particular, quantities such as the overshoot/undershoot of
the output signal (maximum positive/negative deviation from a
reference value) and the settling time (time it takes for transient
behaviors to converge to some small region around the reference value)
are of interest. Given a settling time and overshoot for the first
system, we would like the second system to display similar
characteristics. We remark that both of these properties can be
expressed in \STL, see \cite{hscc14benchmark} for details.  We
quantify system conformance (and thereby adherence to requirements) in
terms of the Skorokhod distance, or, in other words, maximum permitted
time/space-jitter value $\delta$. For this system, we know that at
nominal conditions, the settling time is approximately $2.5$ seconds,
and that we can tolerate an increase in settling time of about $0.5$
seconds. Thus, we chose a time-saling factor of $2 = \frac{1}{0.5}$.
We observe that the range of $\theta$ is about $0.4$ radians, and
specify an overshoot of $20\%$ of this range as being permissible.
Thus, we pick a scaling factor of $0.08$ for the signal domain. In
other words, Skorokhod distance $\delta = 1$ corresponds to either a
time-jitter of $0.5$ seconds, or a space-discrepancy of $0.08$
radians.

We summarize the results of conformance testing for different values
of sampling time $\Delta$ in Table~\ref{table:example1}.  It is clear
that the conformance of the systems decreases with increasing $\Delta$
(which is to be expected). The time taken to compute the Skorokhod
distance decreases with increasing $\Delta$, as the number of
time-points in the two traces decreases.

\begin{table}[t] \centering 
\caption{Variation in Skorokhod Distance with changing sampling time for an
aircraft pitch control system with an LQR-based controller. Time taken
indicates the total time spent in computing the upper bound on the
Skorokhod distance across all simulations.  We scale the signals such that
a time-jitter of $0.5$ seconds, is treated the same as a value-difference
of $0.08$ radians, and the window size chosen is $150$. The system is
simulated for $5$ seconds, with a variable-step solver.
\label{table:example1}}

\begin{tabular*}{.8\textwidth}{@{\extracolsep{\fill}}lrrr}
\toprule
Controller   &  Skorokhod & Time taken (seconds)         & Number of  \\
Sample-Time  &  distance  & to compute $\dist_{\skoro}$  & simulations \\
(seconds)    &            &                              & \\
\midrule
0.01         &   0.012    & 232                          & 104\\
0.05         &   0.049    &  96                          & 104\\
0.1          &   0.11     &  70                          & 106\\
0.3          &   0.39     &  45                          & 104\\
0.5          &   1.51     &  40                          & 101\\
\bottomrule
\end{tabular*}
\end{table}

\mypara{Case Study: Air-Fuel Ratio Controller} In
\cite{hscc14benchmark}, the authors present three systems representing
an air-fuel ratio ($\lambda$) controller for a gasoline engine, that
regulate $\lambda$ to a given reference value of $\lambda_{\text{ref}} =
14.7$.  Of interest to us are the second and the third systems.  The
former has a continuous-time plant model with highly nonlinear
dynamics, and a discrete-time controller model.  In
\cite{hscc14lyapunov}, the authors present a version of this system
where the controller is also continuous.  We take this to be
$\system_1$. The third system in \cite{hscc14benchmark} is a
continuous-time closed-loop system where all the system differential
equations have right-hand-sides that are polynomial approximations of
the nonlinear dynamics in $\system_1$. We call this polynomial
dynamical system $\system_2$. The rationale for these system versions
is as follows: existing formal methods tools cannot reason about
highly nonlinear dynamical systems, but tools such as Flow*
\cite{flowstar}, C2E2 \cite{c2e2}, and CORA \cite{althoff} demonstrate
good capabilities for polynomial dynamical systems. Thus, the hope is
to analyze the simpler systems instead. In \cite{hscc14benchmark}, the
authors comment that the system transformations are not accompanied by
formal guarantees.  By quantifying the difference in the system
behaviors, we hope to show that if the system $\system_2$ satisfies
the temporal requirements $\varphi$ presented in
\cite{hscc14benchmark}, then $\system_1$ satisfies a moderate
relaxation of $\varphi$. We pick a scaling factor of $2$ for the time
domain, as a time-jitter of $0.5$ seconds is the maximum deviation we
wish to tolerate in the settling time, and pick $0.68 =
\frac{1}{0.1*\lambda_{\text{ref}}}$ as the scaling factor for $\lambda$
(which corresponds to the worst case tolerated discrepancy in the
overshoot).

The results of conformance testing for these systems are summarized in
Table~\ref{table:example2}. In \cite{arch14benchmark}, the authors posed a
challenge problem for conformance testing. In it, the authors
reported that the original nonlinear system and the approximate polynomial
system both satisfy the \STL requirements specifying overshoot/undershoot
and settling time.  We, however, found an input that causes the outputs of
the two systems to have a high Skorokhod distance.  Thus, comparing the two
systems by considering equi-satisfaction of a given set of \STL
requirements such as overshoot/undershoot and settling time may not always
be sufficient, and our experiment indicates that the more nuanced Skorokhod
metric may be a better measure of conformance. 


\begin{table}[t] 
\centering 
\caption{Conformance testing for closed-loop A/F ratio controller at
different engine speeds. We scale the signals such that 0.5 seconds of
time-jitter is treated equivalent to 10\% of the steady-state value (14.7)
of the A/F ratio signal. The simulation traces correspond to a time horizon
of $10$ seconds, and the window size is $300$.  \label{table:example2}}

\begin{tabular*}{.8\textwidth}{@{\extracolsep{\fill}}lrrrr}
\toprule
Engine       & Skorokhod & Computation & Total Time   & Number of \\
speed (rpm)  & distance  & Time (secs) & Taken (secs) & simulations \\
\midrule
1000         & 0.31      & 218         & 544           &  700 \\
1500         & 0.20      & 240         & 553           &  700 \\
2000         & 0.27      & 223         & 532           &  700 \\
\bottomrule
\end{tabular*}
\end{table}

\mypara{Case Study: Engine Timing Model} The Simulink demo palette
presented by the Mathworks \cite{mathworks_simulink_demo} contains a
system representing a four-cylinder spark ignition internal combustion
engine based on a model by Crossley and Cook
\cite{crossley1991nonlinear}. This system is then enhanced by adding a
proportional plus integral (P+I) control law. The integrator is used
to adjust the steady-state throttle as the desired engine speed
set-point changes, and the proportional term compensates for phase lag
introduced by the integrator. In an actual implementation of such a
system, such a P+I controller is implemented using a discrete-time
integrator. Such integrator blocks are typically associated with a
particular numerical integration technique, {\em e.g.}, forward-Euler,
backward-Euler, trapezoidal, {\em etc}. It is expected for different
numerical techniques to produce slight variation in the results, and
we wish to quantify the effect of using different numerical
integrators in a closed-loop setting.  We try to check if the
user-provided bound of $\delta = 1.0$ is satisfied by systems
$\system_1$ and $\system_2$, where $\system_1$ is the original system
provided at \cite{mathworks_simulink_demo}, while $\system_2$ is a
modified system that uses the backward Euler method to compute the
discrete-time integral in the controller. We try to determine the
input signal that leads to a violation of this $\delta$ bound, using a
simulation-guided approach as described before. We scale the outputs
in such a way that a value discrepancy of $1\%$ of the the output
range (~$1000$) is equivalent to a time discrepancy of $0.1$ seconds.
These values are chosen to bias the search towards finding signals
that have a small time jitter. This is an interesting scenario for
this case study where the two systems are exactly equivalent except
for the underlying numerical integration solver.  We find the signal
shown in Fig.~\ref{fig:violation}, for which we find output traces
with Skorokhod distance $1.04$. The experiment uses $296$ simulations
and the total time taken to find the counterexample is $677$ seconds.

\begin{figure}[t]
        \centering
        \scalebox{0.9}{
        \includegraphics[trim=15mm 87mm 0mm 82mm,clip,width=100mm]{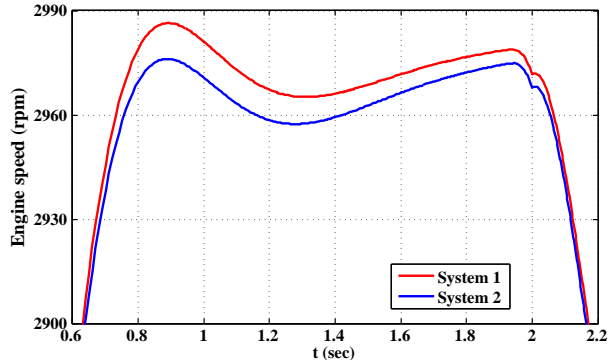}}
\vspace{-2mm}
\caption{Example of non-conformant behavior found using a simulation-guided
optimization algorithm with the Skorokhod distance between system output
trajectories as the cost function.}
\label{fig:violation}
\vspace{2mm}
\end{figure}

\vspace{-3mm}
\section{Conclusion}
\vspace{-2mm}

Metrics for comparing behaviors of dynamical systems 
which quantify both time and value distortions
have heretofore been
an object of mathematical inquiry, without enough attention being paid to
computational aspects and connections to logical requirements.  We argue
that the Skorokhod metric provides a robust definition of conformance by
proving  transference of a rich class of temporal logic properties.
We also demonstrate the computationally tractability of the metric  for practical use 
by  constructing  a conformance testing tool
in a  simulation and optimization guided
approach for  finding and quantifying non-conformant behavior of
dynamical systems.
Pinpointing the source of trace deviations is 
necessary in many engineering applications;
our tool allows for independent weighing of time and value-dimension distortions
in order to achieve this objective.

\newpage
\renewcommand{\baselinestretch}{1.0}
 \bibliographystyle{plain}
 \bibliography{skoro}

\newpage
\section*{Appendix}

\subsection*{A. Transference Formalism and Proofs}
\begin{example}[Freeze Quantification]
Suppose we want to express that whenever the event $Q$ occurs, it is followed
later by $R$, and then by $S$, such that the time difference between occurrences 
of $Q$ and $R$ is at most $5$, and also the time difference between occurrences 
of $Q$ and $S$ is at most $10$.
This can be expressed in  \mTLTL{\small$(\TFun)$} as
\[\Box\Big( x.Q \rightarrow \Diamond \big(y.\big[ R\wedge(y\leq x+5) \wedge
\Diamond \left(z. \left(S \wedge z\leq x+10\right)\right)\big]\big) \Big).\]
Thus, freeze quantification,  by giving a mechanism to bind times to variables,
allows us to relate, 
with several constraints, far apart events.
\qed
\end{example}

\begin{example}[Freeze Quantification Functions]
Suppose we want to express that  whenever the event $Q$ occurs, 
it must be followed by a response $R$ within time $\lambda^{t_Q}$ for some 
$\lambda > 1$ where $t_Q$ is the time at which $Q$ occurred; 
thus, the later
$Q$ occurs the more delay we can tolerate in the response time.
The requirement can be expressed as
$x. \left( Q \rightarrow \Diamond 
\big(y. \left(R \wedge 0 \leq y\leq \lambda^x\right)\big)\right)$.
\qed
\end{example}

\begin{example}[$\delta$-relaxation for Bounded Temporal Operators -- \MTL]
\label{example:RelaxationMTL}
We demonstrate how $\delta$-relaxation operates on bounded time constraints
through an example.
Consider an \MTL formula $\phi= Q \until_{[a,b]} R$.
The  $\delta$-relaxation of this formula over the closed interval ${ I_{\TFun}} = \reals_+$
is   $Q \until_{[a-2\cdot\delta\,,\, b+2\cdot \delta]} R$.
This can be seen as follows.
The formula $\phi$ can be written in \TLTL syntax as:
\[
x. Q \until y. \left( (y\leq x+b) \wedge (y \geq x+b) \wedge R\right).\]
The $\delta$-relaxation of this formula according to
Definition~\ref{definition:Relaxation} is:
\begin{align*}
  \myrelax_{\reals_+}^{\delta}
\left(x. Q \until y. \left( (y\leq x+b) \wedge (y \geq x+a) \wedge R\right) \right) \ & =\\
&\hspace{-30mm} = \myrelax_{\reals_+}^{\delta}
\left(x. Q \until y. \left( (y-x-b\leq 0) \wedge (y -x-a\geq 0) \wedge R\right) \right) \\
& \hspace{-30mm} =
x. Q \until y. \left( 
\begin{array}{l}
(y-x-b - 2\vdot \delta\leq0 ) \ \wedge\\
 (y -x -a +2\vdot \delta\geq 0) \wedge R
\end{array}\right) 
\\
&\hspace{-20mm}
\begin{array}{l}
\text{ since the Lipschitz constant of } y-x-c  \text{ is } 2 \\
\text{ for any constant } c
\end{array}\\
&\hspace{-30mm}  = x. Q \until y. \left( (y \leq x+b+2\vdot\delta ) \wedge (y \geq x +a -2\vdot \delta) \wedge R\right) \\
& \hspace{-30mm} = Q \until_{[a-2\cdot\delta,b+2\cdot \delta]} R.
\end{align*}
Thus, the time constraint interval boundaries are relaxed by $2\vdot \delta$.
The factor of $2$ arises because there are two
contributing factors: the starting time of $Q$ can be ``pulled back'' by 
$\delta$, and the time of $R$ can be postponed by $\delta$; thus, the time
duration
in between $Q$ and $R$ increases by $2\vdot \delta$.
\qed
\end{example}

\smallskip\noindent\textbf{Removing Negation using the $\awaits$ Operator.}
The following  identities hold relating the $\awaits$ operator  to the
$\until$ operator 
\begin{compactenum}
\item 
$ \phi_1 \until \phi_2\, \equiv\,  \neg\left(\neg(\phi_2) \awaits (\neg\phi_1 \wedge \neg\phi_2)\, \right)$; and
\item 
$ \phi_1 \awaits \phi_2 \,  \equiv \,  \neg\left(\neg(\phi_2) \until (\neg\phi_1 \wedge \neg\phi_2)\, \right)$.
\end{compactenum}
Informally, the first identity states that $\neg (\phi_1 \until \phi_2)$ holds iff either
(i)~$\phi_2$ never holds; or
(ii)~there is a point where $\phi_1$ is false, and at that point and all points before it, $\phi_2$ has remained false.
The second identity is similar.
The first identity above allows us to ``push'' the negations down using the $\awaits$ operator.
The mechanism for the three interesting cases is below.
\begin{align*}
\neg\left(f_{\mytime}(x_1, \dots, x_l) \,  \sim\, 0 \right)   & 
\ \equiv\  f_{\mytime}(x_1, \dots, x_l) \, \myneg(\sim)\,  0,\\
& \hspace{2.5cm}
\text{where, for  }  \sim \,\in \set{\leq, <, \geq, >} \text{ we have}\\
&\hspace{2.5cm} \begin{aligned}
\myneg(\leq) & \text{ to be } >; \hspace{5mm} &
\myneg(<) &  \text{ to be }   \geq;\\
\myneg(\geq) &  \text{ to be }   <; \hspace{5mm} & 
\myneg(>) &  \text{ to be }   \leq
\end{aligned}\\
\neg(x.\psi) \, & \equiv\,  x.\neg(\psi)\\    
\neg\left(\phi_1 \until \phi_2 \right) \, &  \equiv\,  \neg(\phi_2) \awaits (\neg\phi_1 \wedge \neg\phi_2)
\end{align*}

\begin{proposition}
\label{proposition:PropositionalRelaxation}
The function $\myrelax$ is a relaxation on   \mTLTL{\small$(\TFun)$} formulae,
\emph{i.e.} if a timed propositional trace $\pi\models \phi$ for
a  \mTLTL{\small$(\TFun)$} formula $\phi$,
then $\pi\models \myrelax_{ I_{\TFun}}^{\delta}(\phi)$.
\end{proposition}
\begin{proof}
Observe that, over the predicates $f_{\mytime}(x_1, \dots, x_l)   \sim 0$, the
function $\myrelax$ is indeed a relaxation, \emph{i.e}, 
if $f_{\mytime}(t_1, \dots, t_l)  \sim 0$ for
values $t_1, \dots, t_l$, then 
$\myrelax_{ I_{\TFun}}^{\delta}\left(f_{\mytime}(t_1, \dots, t_l)  \right) \sim 0) $ also holds.
The result follows by a straightforward induction argument.\qedhere
\end{proof}

\begin{proof}[\textbf{Proof of 
Theorem~\ref{theorem:PropositonalRobustness}}]
Let $\untime(\phi)$ be the formula where all freeze variable  constraints are replaced
by $\true$ (\emph{e.g.}  $\untime(x. (Q \wedge x<5)) $ is $x. (Q\wedge \true)$).
Since $\dist(\pi, \pi') < \delta$, we have that there exists a retiming
$r: \tdom(\pi) \mapsto \tdom(\pi')$ such that 
\begin{equation}
\label{equation:SameSequence}
\pi(t) = \pi'(\retime(t)).
\end{equation}
This implies that
both $\pi$ and $\pi'$  satisfy $\untime(\phi)$, which can be
shown by an induction argument.
The interesting cases are for the $\until$ and $\awaits$ operators.
We sketch the argument for the $\until$ case (the argument for
$\awaits$ is similar).
The time environment $\env'$ for $\pi'$ assigns the time $\retime(t_x) $
to the freeze variable $x$ where the witnessing freeze variable environment $\env$ for
$\pi\models \phi$ assigns $t_x$ to $x$.
Let  $\pi \models_{\env} \phi_1 \until \phi_2$, and let
$t$ be the time value which demonstrates this satisfaction (as in 
Definition~\ref{definition:PropositionalSemantics}), with the corresponding
freeze variable environment $\env$.
To show $\pi' \models_{\env'} \phi_1 \until \phi_2$, we pick the time
$\retime(t)$, with the environment $\env'$  for $\pi'$ which 
 assigns the time $\retime(t_x) $
to the freeze variable $x$ where $\env(x) = t_x$.
It can be checked that, due to Equation~\ref{equation:SameSequence},
we have
(i)~$\retime(t) \geq  \env'(x)$, for a freeze variable $x$ in 
 $\phi_1 \until \phi_2$
 (which was previously bound);
(i)~${\pi'}^{\retime(t)} \models_{\env'} \phi_2$; and
(ii)~for all $t_0' \leq t^{\dagger} < \retime(t)$, we have 
${\pi'}^{t^{\dagger}} \models_{\env'} \phi_1\vee \phi_2$.
Thus, $\retime(t)$, and $\env'$ demonstrate that
$\pi'\models_{\env'} \phi_1 \until \phi_2$.

We now check what is the relaxation needed on the original freeze variable constraints
so that $\pi'$  satisfies the relaxed constraints.
Without loss of generality, assume that each freeze variable $x$ is only quantified once
in $\phi$, \emph{i.e.} once it is bound to a value by ``$x.$'', it is not ``re-bound'' with
another application of ``$x.$''.

Let $\kappa_{\pi}$ denote an assignment of time values (from $I$)
 to the freeze variables
such that all the freeze variable constraints in $\phi$ are satisfied, \emph{i.e.} $\kappa_{\pi}$
is an time environment witness to the satisfaction of $\phi$ by $\pi$.
Consider a free variable assignment $\kappa_{\pi'}$ corresponding to
$\kappa_{\pi}$, where  $\kappa_{\pi'}(x) = \retime\left(\kappa_{\pi}(x)\right)$.
This is a legal variable assignment compatible with some $\until, \awaits$ time
 witnesses which demonstrate that
$\pi'$ satisfies
 $\untime(\phi)$, as shown previously.
Observe that by the existence of a retiming function, for all freeze variables  $x$ occurring in $\phi$, we have
that $\abs{\kappa_{\pi'}(x) -\kappa_{\pi}(x)} < \delta$.

Since the time \emph{values} of variables are different in $\kappa_{\pi}$
and $\kappa_{\pi'}$, the original freeze constraints (\emph{e.g.} $x<5$) in $\phi$
might not be satisfied with the assignment $\kappa_{\pi'}$.
Consider a freeze variable constraint $f_{\mytime}(x_1, \dots, x_l) \sim 0$ in $\phi$.
We know that $f_{\mytime}(\kappa_{\pi}(x_1), \dots, \kappa_{\pi}(x_l)) \sim 0$ is true.
As $\abs{\kappa_{\pi'}(x) -\kappa_{\pi}(x)} \leq \delta$ for all 
 freeze variables  $x$ occurring in $\phi$, 
by the definition of relaxation, we have that 
\begin{compactenum}
\item 
$f_{\mytime}(\kappa_{\pi}(x_1), \dots, \kappa_{\pi}(x_l)) +  K_{\mytime} ( \delta) \sim 0$
if $\sim\, \in \set{>, \geq}$; and 
\item 
$f_{\mytime}(\kappa_{\pi}(x_1), \dots, \kappa_{\pi}(x_l))  - K_{\mytime} (\delta) \sim 0$
if $\sim\, \in \set{<, \leq}$.
\end{compactenum}
This implies that $\kappa_{\pi'}$ is also a witness to the satisfaction of 
$\myrelax_{I_{\pi,\pi'}}^{\delta}(\phi)$ by $\pi'$.
Thus, $\pi'\models \myrelax_{I_{\pi,\pi'}}^{\delta}(\phi)$.
\qedhere
\end{proof}

\begin{proof}[\textbf{Example~\ref{example:Transference-one} details}]
Since $\pi$  satisfies $\phi$, we must have time-stamps $t_x, t_y, t_z$ bound to
$x,y,z$ respectively so that with these assignments, the formula $\phi$
is satisfied.
Since  $\pi'$ is  $\delta$ close to $\pi$, for  every $\epsilon > 0$, there is a
retiming from $\pi$ to $\pi'$ such that the times $t_x, t_y, t_z$ in $\pi$
are mapped to  $t_x', t_y', t_z'$ in $\pi'$ such that
(a)~$|t_x-t_x'| \leq \delta+\epsilon$; and
(b)~$|t_y-t_y'| \leq \delta+\epsilon$; and
(c)~$|t_z-t_z'| \leq \delta+\epsilon$.
Let $\delta'= \delta+\epsilon$.

The sum $(t_x'-t_y')^2 + (t_y'-t_z') ^2+(t_z'-t_x')^2$ is
\begin{align*}
& =\,  \left( (t_x'-t_x) + (t_x-t_y) + (t_y-t_y')\right)^2 \, +\, 
 \left( (t_y'-t_y) + (t_y-t_z) + (t_z-t_z')\right)^2 \, +\\
& \hspace*{63mm} \left( (t_z'-t_z) + (t_z-t_x) + (t_x-t_x')\right)^2\\
& =\, 
2 \left((t_x'-t_x)^2  + (t_y'-t_y)^2  + (t_z'-t_z)^2\right) \,  +\,
 (t_x-t_y)^2 + (t_y-t_z) ^2+(t_z-t_x)^2\, +\,\\
& \hspace*{25mm}
2\left( (t_x'-t_x)(t_x-t_y) +  (t_y-t_y')(t_x-t_y)+ (t_x'-t_x)(t_y-t_y')\right)
\, +\\
&\hspace*{25mm}
2\left( (t_y'-t_y)(t_y-t_z) +  (t_z-t_z')(t_y-t_z)+ (t_y'-t_y)(t_z-t_z')\right)
\, +\\
& \hspace*{25mm}
2\left( (t_z'-t_z)(t_z-t_x) +  (t_x-t_x')(t_z-t_x)+ (t_z'-t_z)(t_x-t_x')\right)\\
& \leq \,
6\delta'^2 + d+ 
2\delta'\abs{t_x-t_y} + 2\delta'^2
+2\delta'\abs{t_y-t_z} + 2\delta'^2
+ 2\delta'\abs{t_z-t_x} + 2\delta'^2\\
& = \,
d + 12 \delta'^2 + 4\delta'\left(\abs{t_x-t_y} + \abs{t_y-t_z} + \abs{t_z-t_x} \right)\\
& \leq \,
d + 12 \cdot\delta'^2 + 4\sqrt{3}\cdot\delta'\cdot \sqrt{d}
\end{align*}
In the last step above, we use the inequality:
$\abs{a} + \abs{b} + \abs{c} \leq \sqrt{3}\cdot\sqrt{ a^2 + b^2+c^2}$
This inequality is obtained by applying the Cauchy-Schwarz inequality to
the tuples $(\abs{a}, \abs{b}, \abs{c})$ and $(1,1,1)$.
Thus, by Theorem~\ref{theorem:PropositonalRobustness},
for every $\epsilon >0$, we have
\[\pi'\models
x. \left( Q \rightarrow \Diamond
\big(y. \left(R \wedge
 \Diamond\left[z. \left(S \wedge
\left(  (y-x)^2 + (z-y)^2 + (z-x)^2 \leq d^{\dagger} \right) \right)
\right] \right)
\big)\right) \]
where 
$d^{\dagger} = d + 12 \cdot\delta'^2 + 4\sqrt{3}\cdot\delta'\cdot \sqrt{d}$.
\qedhere
\end{proof}

\begin{definition}[$\delta$-relaxation of \mTLTL{\small$(\TFun, \SFun)$} formulae]
\label{definition:Signal Relaxation}
Let $\phi$ be a  \mTLTL{\small$(\TFun, \SFun)$} formula in which negations appear only on the 
prepositional symbols .
The $\delta$ relaxation of $\phi$ (for $\delta \geq 0$), denoted 
$\myrelax_{I_{\TFun}, \imap}^{\delta}(\phi)$  is defined as follows, 
where $I_{\TFun}$, a closed subset of
$reals_+$, is the domain of the variables in $V_{\mytime}$; 
and $ \imap$ is a mapping from $V_{\mysig}$  to closed intervals
of $\reals$ such that $\imap(z)$ denotes the domain of $z$.
\begin{align*}
\myrelax_{I_{\TFun}, \imap}^{\delta}(\true) & = \true; 
\hspace{30mm} \myrelax_{I_{\TFun}, \imap}^{\delta}(\false)  = \false;\\
\myrelax_{I_{\TFun}, \imap}^{\delta}( \phi_1 \wedge \phi_2 ) & = 
\myrelax_{I_{\TFun}, \imap}^{\delta}( \phi_1)  \wedge \myrelax_{\delta}(\phi_2 );\\
 \myrelax_{I_{\TFun}, \imap}^{\delta}( \phi_1 \vee \phi_2 )  & = 
 \myrelax_{I_{\TFun}, \imap}^{\delta}( \phi_1)  \vee \myrelax_{I_{\TFun}, \imap}^{\delta}(\phi_2 );\\
\myrelax_{I_{\TFun}, \imap}^{\delta}( x.\psi ) & = x.\myrelax_{I_{\TFun}, \imap}^{\delta}(\psi);\\
\myrelax_{I_{\TFun}, \imap}^{\delta}(  \phi_1 \until \phi_2 )&  =
\myrelax_{I_{\TFun}, \imap}^{\delta}(\phi_1) \until \myrelax_{I_{\TFun}, \imap}^{\delta}(\phi_2);\\
\myrelax_{I_{\TFun}, \imap}^{\delta}(  \phi_1 \awaits \phi_2 )&  =
\myrelax_{I_{\TFun}, \imap}^{\delta}(\phi_1) \awaits \myrelax_{I_{\TFun}, \imap}^{\delta}(\phi_2)  
\\
\myrelax_{I_{\TFun}, \imap}^{\delta}\left(f_{U}(z_1, \dots, z_l)  \right) \sim 0)  & = 
\begin{cases}
f_{U}(z_1, \dots, z_l) \,  +\,  K_{f_{U}}(\delta) \  \sim\,  0 & \text{ if } 
 \sim\, \in\set{>, \geq};\\
f_{U}(z_1, \dots, z_l) \, -\,   K_{f_{U}}(\delta) \  \sim \, 0 & \text{ if } 
 \sim \,\in\set{<, \leq};\\
\end{cases}\\
& \qquad \text{ where } U \in \set{ \mytime, \mysig} \text{ with }  K_{f_{U}} 
\text{ being as in Definition~\ref{definition:Relaxation};}\\
&\qquad\text{ and  }
K_{f_{\mysig}}: \big[0,\  \max_{z\in V_{\mysig}} |\max \imap(z) \, - \, \min\imap(z)| \big]
\mapsto \reals_+\\
&\qquad  \text{ is a function such that:}\\
&\hspace{-40mm}
 K_{f_{\mysig}}(\delta) = 
\sup_{
  \begin{array}{c}
    z_i\in \imap(z_i); \,  z'_i\in \imap(z'_i)\\
    \text{ for all } i 
 \end{array}
}
\left\{
\left\arrowvert
  \begin{array}{c}
    f_{\mysig}(z_1,\dots, z_l) \\
    -\\
    f_{\mysig}(z_1',\dots,z_l') 
  \end{array}\right\arrowvert
\text{ s.t. }
  |z_i-z_i'| \leq \delta \text{ for all } i
\right\}\qedhere
\\
\end{align*}
\qed
\end{definition}
The functions  $K_{f_{\mysig}}(\delta)$ define the maximal change in the value
of $f_{\mysig}$ that can occur when the input variables can vary by $\delta$.
The role of $\imap$ in the above definition is to restrict the domain of the signal variables
in order to obtain the least possible bounds relaxation bounds on the signal
constraints; as was done in Definition~\ref{definition:Relaxation} for the
freeze variables.

\begin{proposition}
\label{proposition:SignalTLTLRelaxation}
The function $\myrelax_{I_{\TFun}, \imap}^{\delta}$ is a relaxation on   \mTLTL{\small$(\TFun, \SFun)$} formulae,
\emph{i.e.} if a timed $\reals^n$-valued trace  $\pi\models \phi$ for
a  \mTLTL{\small$(\TFun, \SFun)$} formula $\phi$,
then $\pi\models \myrelax_{I_{\TFun}, \imap}^{\delta}(\phi)$.
\end{proposition}
\begin{proof}
The proof is similar to the proof of Proposition~\ref{proposition:PropositionalRelaxation}.
\qedhere
\end{proof}

\begin{proof}[\textbf{Proof of 
Theorem~\ref{theorem:SignalTLTLRobustness}}]
The proof use the result for the propositional case, Theorem~\ref{theorem:PropositonalRobustness}.
We construct the propositions $p_{f_{\mysig}}$ defined to be
$\myrelax_{I_{\TFun}, \imap}^{\delta}\left(f_{\mysig}(\overline{y}) \right) \sim 0) $ 
for the constraints  over $V_{\mysig}$ 
in the formula $\phi$; and define the  \mTLTL{\small$(\TFun)$} formula
$\phi_{\prop}$ as that obtained from $\phi$ by syntactically replacing
each constraint $ f_{\mysig}(\overline{y}) \sim 0$ in $\phi$ by $p_{f_{\mysig}}$.
Let $\prop_{\mysig}$ denote all such predicates for $\phi$.
We obtain the timed  $\prop_{\mysig}$ propositional traces 
$\pi_{P_{\mysig}}, \pi_{P_{\mysig}}'$
from $\pi, \pi'$ by mapping to propositions.
By the definition of the skorokhod distance, the distance between
 $\pi_{\prop_{\mysig}}$ and $\pi_{\prop_{\mysig}}'$ is less than $\delta$.
By Theorem~\ref{theorem:PropositonalRobustness}, 
$ \pi_{\prop_{\mysig}}' \models \phi_{\prop}$.
This implies  $\pi' \models \myrelax_{I_{\TFun}, \imap}^{\delta}(\phi)$.
\qedhere
\end{proof}

\subsection*{B. Details on Case Studies}

\paragraph{LQR-based pitch controller.} The aircraft pitch controller
system has $3$ state variables, and the state vector $\mathbf{x}$ =
$[\alpha\ q\ \theta]$, where $\alpha$ is the angle of attack, $q$ is the
pitch rate, and $\theta$ is the pitch angle. The system has a single input
$\delta$ (the elevator deflection angle).  In deriving the control law, the
designers use the state feedback law to substitute $\delta$ =
$\theta_{des}-K\mathbf{x}$, where $\theta_{des}$ is the desired pitch
angle. The resulting dynamical equations of the system are of the form
$\dot{\mathbf{x}} = (A-BK)\mathbf{x} + B\theta_{des}$, and the output of
the system is the state variable $\theta$.  Note that the $K$ matrix is the
gain matrix resulting from the LQR control design technique.  The values of
the $A$, $B$ and $K$ matrices are as given below:
\[
\begin{array}{ll}
\vspace{1em}
A = \left[\begin{array}{l@{\hspace{1em}}l@{\hspace{1em}}l}
          -0.313  & 56.7   & 0 \\
          -0.0139 & -0.426 & 0 \\
          0       & 56.7   & 0 \\
          \end{array}
    \right] &
B = \left[\begin{array}{l} 
          0.232 \\
          0.0203 \\
          0 \\
          \end{array}
     \right] \\
K = [-0.6435\ \ 169.6950\ \ 7.0711] &  \\
\end{array}
\]

\paragraph{Air-Fuel Ratio Controller.} The Air-Fuel (A/F) ratio control
systems that we consider are simplified versions of industrial-scale
models. Both versions have $2$ exogenous inputs, and $4$ continuous states.
The inputs are engine speed (measured in rpm) and the throttle angle (in
degrees).  The throttle angle is a user input, and it is common to assume a
series of pulses or steps as throttle angle inputs. The engine speed is
considered an input to avoid modeling parts of the powertrain dynamics. In
our experiments, we typically hold the engine speed constant. This is to
mimic a common engine testing scenario involving a dynamometer, which is a
device to provide external torque to the engine to maintain it at a
constant speed. Of the $4$ continuous states, we assume that $2$ of these
states are from the plant model (that encapsulates physical processes
within the engine), while $2$ states belong to the controller. The plant
states $p$ and $\lambda$ denote intake manifold pressure and the A/F ratio
respectively. The controller states $p_e$ denotes the estimated manifold
pressure (with the use of an observer) used in the feed-forward control,
and the state $i$ denotes the integrator state in the P+I feedback control.
We check conformance with respect to the system output $\lambda$.  For the
dynamical system equations, please refer to
\cite{hscc14benchmark,hscc14lyapunov}.

\end{document}